\documentclass[10pt,a4paper]{article}

\usepackage{amstext,amsgen,latexsym}
\usepackage{amstext,amssymb,amsfonts,latexsym}
\usepackage{theorem}
\usepackage{pifont}
\usepackage[dvips]{graphics,epsfig}

 \newcommand{\bs}{\bigskip}
 \newcommand{\ms}{\medskip}
 
 \newcommand{\s}{\smallskip}
 \newcommand{\hs}[1]{\hspace*{ #1 mm}}
 \newcommand{\vs}[1]{\vspace*{ #1 mm}}

 \newcommand{\setempty}{\mathrm{\O}}
 \newcommand{\real}{\mathbb{R}}

 \newcommand{\nat}{\mathbb{N}}
 
 \newcommand{\integer}{\mathbb{Z}}

 \newcommand{\prob}{{\mathrm{Prob}}}

 \newcommand{\co}{\mathrm{co}\mbox{-}}

 \newcommand{\ie}{\textrm{i.e.},\hspace*{2mm}}
 \newcommand{\eg}{\textrm{e.g.},\hspace*{2mm}}
 
 \newcommand{\etalc}{\textrm{et al.}}

 \newcommand{\PP}{{\cal P}}

 \newcommand{\p}{\mathrm{P}}

 \newcommand{\bpp}{\mathrm{BPP}}
 \newcommand{\pp}{\mathrm{PP}}

 \newcommand{\cequalp}{\mathrm{C}_{=}\mathrm{P}}

 \newcommand{\oneplin}{1\mbox{-}\mathrm{PLIN}}
 \newcommand{\onecequallin}{1\mbox{-}\mathrm{C}_{=}\mathrm{LIN}}

 \newcommand{\onebplin}{1\mbox{-}\mathrm{BPLIN}}

 \newcommand{\onedlin}{1\mbox{-}\mathrm{DLIN}}

 \newcommand{\reg}{\mathrm{REG}}
 \newcommand{\cfl}{\mathrm{CFL}}
 \newcommand{\dcfl}{\mathrm{DCFL}}

 \newcommand{\matrices}[4]{\left( \begin{array}{cc} #1 & #2 \\%
      #3 & #4   \end{array}\right)}

 \newcommand{\ninematrices}[9]{\left( \begin{array}{ccc} 
      #1 & #2 & #3 \\%
      #4 & #5 & #6 \\%
      #7 & #8 & #9    \end{array}\right)}
      
\theoremstyle{plain}
\theoremheaderfont{\bfseries}
 \newtheorem{theorem}{Theorem}[section]
 \newtheorem{lemma}[theorem]{Lemma}
 \newtheorem{proposition}[theorem]{Proposition}
 \newtheorem{corollary}[theorem]{Corollary}

 {\theorembodyfont{\rmfamily}
  }
 {\theorembodyfont{\rmfamily} }
 {\theorembodyfont{\rmfamily} }

 \newtheorem{claim}{Claim}

 \newenvironment{proof}{\par \noindent
            {\bf Proof. \hs{2}}}{\hfill$\Box$ \vspace*{3mm}}

 \newenvironment{proofof}[1]{\vspace*{5mm} \par \noindent
         {\bf Proof of #1.\hs{1}}}{\hfill$\Box$ \vspace*{3mm}}

\newif\ifnotesw\noteswtrue
   {\ifnotesw\marginpar[\hfill\(\top\)]{\(\top\)}\fi}%
      {\ifnotesw\marginpar[\hfill\(\bot\)]{\(\bot\)}\fi}
      
\newcommand{\mnote}[1]%
   {\ifnotesw\marginpar%
	  [{\scriptsize\begin{minipage}[t]{\marginparwidth}
	  \raggedleft#1%
		  \end{minipage}}]%
	  {\scriptsize\begin{minipage}[t]{\marginparwidth}
	  \raggedright#1%
		  \end{minipage}}%
    \fi}

\newcommand{\ignore}[1]{}

\newcommand{\track}[2]{[{\tiny \begin{array}{c} #1 \\%
      #2 \end{array} }]}
\newcommand{\cent}{{|}\!\!\mathrm{c}}
\newcommand{\dollar}{\$}
\newcommand{\averreg}{\mathrm{Aver\mbox{-}REG}}

\begin{document}
 \pagestyle{plain}
 
\begin{center}
{\Large {\bf The Roles of Advice to One-Tape Linear-Time \s\\
Turing Machines and Finite Automata}}\footnote{An extended abstract  appeared in the Proceedings of the 20th International Symposium on Algorithms and Computation (ISAAC 2009), Lecture Notes in Computer Science, Springer-Verlag, Vol.5878, pp.933--942, December 16--18, Hawaii, USA, 2009.} \bs\\
{\sc Tomoyuki Yamakami}\footnote{Current Affiliation: 
Department of Information Science, University of Fukui, 
3-9-1 Bunkyo, Fukui, 910-8507 Japan} \bs\\
\end{center}


\begin{quote}
{\bf Abstract.}\hs{1}
We discuss the power and limitation of various ``advice,'' when it is given particularly to weak computational models of one-tape linear-time Turing machines and one-way finite (state) automata. Of various advice types, we consider deterministically-chosen advice (not necessarily algorithmically determined) and randomly-chosen advice (according to certain probability distributions). 
In particular, we show that certain weak machines can be significantly enhanced in computational power when randomized advice is provided  
in place of deterministic advice. 

\ms


\end{quote}
\ms

\section{Advice and Weak Computational Models}

When a machine has a clear, limited operational capability, how can we enhance its computational power beyond its plausible limitation? A straightforward  way is to provide a piece of supplemental external information besides original input data so that such extra knowledge helps the machine solve a target problem efficiently. 
A notion of so-called {\em advice} is such additional information,  
 which depends only on the size of inputs, given to the underlying machine. 
Since Karp and Lipton \cite{KL82} initiated it in early 1980s, the study of advice has attracted numerous researchers in the fields of, \eg computational complexity and cryptography. 
To grip a better understanding of the roles of the advice, we intend 
to take  a rather simple but direct approach toward an investigation 
of the strengths and limitations of the advice, particularly on weak 
models of advised computations. 

{\em One-tape (two-way one-head) Turing machines} (or 1TMs, in short) running in linear time could be one of the most basic types of computational models ever discussed in  computational complexity theory. 
A theory of linear-time 1TMs has been studied intermittently since mid 1960s (see \cite{TYL04} for references). An immediate advantage of studying such weak models is that we can prove anticipated class separations without relying on any {\em unproven} assumption, such as the existence of one-way functions. Moreover, we can conduct a precise analysis of advice when its underlying computation is limited in power. 
As were shown in \cite{Hen65,Kob85,TYL04}, certain variants of this 1TM 
model are closely tied to {\em one-way finite (state) automata} with 
constant memory space, which could be viewed to run  
a simple form of memoryless online algorithms. 
Advised computations of one-way deterministic finite automata (or 1dfa's, in short) were initially studied in \cite{DH95,TYL04} and deterministic linear-time 1TMs with advice were discussed in \cite{TYL04}. Interestingly, it was shown in \cite{TYL04} that deterministic linear-time 1TMs that take linear-size advice are no more powerful than 1dfa's with advice of size equal to input size. This characterization makes it easier for us to handle the models of linear-time 1TMs. Recently, a series of studies \cite{Yam08a,Yam09a} revealed the power and limitation of advice, when given to its underlying finite automata. 
(Another direction with advice was recently indicated in \cite{AF10}.) 
In addition to standard (deterministic) advice, we also study in this paper  {\em randomized advice} (in which each advice string is chosen at random according to a certain probability distribution), which may allow  its underlying machines to err with, \eg  bounded-error probability 
(\ie at most a certain constant probability away from $1/2$). 
A piece of such randomized advice gives a significantly high power to the underlying machines.  

Concerning the aforementioned models of linear-time 1TMs, we shall 
focus our study only on the following four language families: $\onedlin$ (deterministic), $\onebplin$ (bounded-error probabilistic), $\oneplin$ (unbounded-error probabilistic), and $\onecequallin$ (error probability exactly 1/2), introduced in \cite{TYL04}. These language families can be viewed as ``scaled-down'' versions of the well-known complexity classes, $\p$, $\bpp$, $\pp$, and $\cequalp$. Some of their advised counterparts are succinctly denoted as $\onedlin/lin$, $\oneplin/lin$, and $\onecequallin/lin$. Moreover, for given randomized advice, we write some of their corresponding families as $\onebplin/Rlin$, $\oneplin/Rlin$, and  $\onecequallin/Rlin$. Similarly, based on the finite automata models, we define $\reg/n$ and $\cfl/n$ respectively as the families of regular languages with advice 
and of context-free languages with advice. We further introduce two additional  language families equipped with randomized advice: 
$\reg/Rn$ and $\cfl/Rn$.

In this paper, we shall present new collapses and separations among the above-mentioned advised language families. Our results are summarized 
in Figure~\ref{fig:hierarchy}. To obtain these results, we shall show new characterizations of advised families and also their new structural properties, which are interesting on their own right. We hope that this paper opens a door to a rich research area that sits between computational complexity theory and formal language and automata theory. 

\begin{figure}[t]
\begin{center}
\includegraphics*[width=7.0cm]{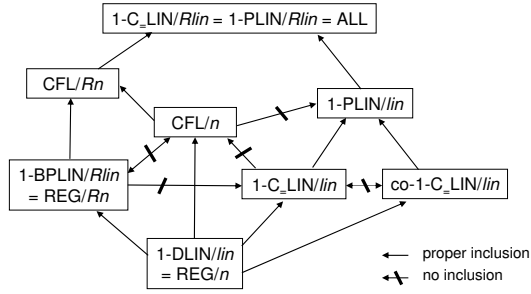}
\caption{A hierarchy of advised language families}\label{fig:hierarchy} 
\end{center}
\end{figure}

\section{Basic Notions and Notations}\label{sec:basic}

We briefly explain fundamental notions and notations used in the subsequent sections. 
Let $\nat$ be the set of all nonnegative integers and set $\nat^{+} = \nat-\{0\}$. For any pair $m,n\in\nat$ with $m\leq n$, the notation $[m,n]_{\integer}$ denotes the integer interval $\{m,m+1,\ldots,n\}$. Conventionally, we write $[m]$ for $[1,m]_{\integer}$. 
Let $\real^{\geq0}$ be the set of all nonnegative real numbers. 
A function $f:\nat\rightarrow\real^{\geq0}$ is said to be 
{\em negligible} if $f(n)\leq 1/p(n)$ for any non-zero polynomial 
$p$ and for all but finitely-many numbers $n$ in $\nat$. 
An {\em alphabet} is a nonempty finite set $\Sigma$ and a {\em string} over $\Sigma$ is a finite series of symbols taken from $\Sigma$. Let $\lambda$ express the {\em empty string}. 
Given a string $x$ over $\Sigma$ and a symbol $\sigma\in\Sigma$, the notation $\#_{\sigma}(x)$ denotes the number of all 
occurrences of $\sigma$ in $x$. The notation $x^{R}$ expresses the string $x$ in reverse.  
The set of all strings over $\Sigma$ is denoted $\Sigma^*$ and a 
{\em language} over $\Sigma$ is a subset of $\Sigma^*$.  For notational convenience, for any language $S$, we define $S(x)=1$ if $x\in S$; 
$S(x)=0$ if $x\not\in S$. The notation $\overline{S}$ for the 
language $S$ is 
the {\em complement} of $S$; namely, 
$\overline{S} = \Sigma^*\setminus S$.   
The {\em length}  of a string $x$, denoted $|x|$, is the total number of occurrences of symbols in $x$. A {\em length function} is a map from $\nat$ to $\nat$. For any length $n\in\nat$, let 
$\Sigma^n = \{x\in\Sigma^*\mid |x|=n\}$. 
A {\em probability ensemble} $\mu$ over $\Sigma^*$ is an infinite series $\{\mu_n\}_{n\in\nat}$, in which each $\mu_n$ is a probability distribution over $\Sigma^n$ (\ie $0\leq \mu_{n}(x)\leq 1$ for any $x\in\Sigma^n$ with $\sum_{x\in\Sigma^n}\mu_n(x)=1$). 

Our basic model of computation is {\em one-tape (or single-tape) two-way one-head off-line Turing machines} (or 1TMs), each of which can be expressed as a sextuple $(Q,\Sigma,\delta,q_0,Q_{acc},Q_{rej})$, where $Q$ is a finite set of {\em inner states}, $\Sigma$ is an alphabet, $q_0$ ($\in Q$) is the {\em initial state}, $\delta$ is a {\em transition function}, $Q_{acc}$ ($\subseteq Q$) is a set of {\em accepting states}, and $Q_{rej}$ ($\subseteq Q$) is a set of {\em rejecting states}. 
Write $Q_{non}$ for the set $Q \setminus (Q_{acc}\cup Q_{rej})$.  This machine $M$ is equipped with one input/work tape, on which an input 
string is initially written, surrounded by two {\em endmarkers} $\cent$ and $\dollar$,  
and a tape head either moves in both directions (to the left or to the right) 
or stays still, after starting from the left endmarker $\cent$. 
For our convenience, let $\check{\Sigma} = \Sigma\cup\{\cent,\dollar\}$.  
When $\delta$ is of the form $\delta:Q_{non}\times\check{\Sigma}\rightarrow Q\times\check{\Sigma}\times\{L,N,R\}$, $M$ is called 
{\em deterministic}. To the contrary, $M$ is {\em probabilistic} if $\delta$ satisfies  $\delta:Q_{non}\times\check{\Sigma}\rightarrow \PP(Q\times\check{\Sigma}\times\{L,N,R\})$, where $\PP(A)$ is the power set of $A$. In this case, each move $(q,\sigma)\mapsto (q'\sigma',d)\in\delta(q,\sigma)$ of $M$ is associated with a {\em transition probability}. If $M$ is deterministic (probabilistic, resp.), 
then we succinctly call it {\em 1DTM} ({\em 1PTM}, resp.). 
The {\em extended transition function} $\hat{\delta}$ induced from $\delta$ is recursively defined as $\hat{\delta}(q,\lambda)= q$ and $\hat{\delta}(q,x\sigma) = \delta(\hat{\delta}(q,x),\sigma)$ for 
each symbol $\sigma\in\Sigma$ and each string $x\in\Sigma^*$.

By adopting Michel's {\em strong definition} for machine's running time \cite{Mic91}, we say that $M$ runs in {\em linear time} if the longest computation path (even in a case of probabilistic computations) of $M$ 
on any input $x$ of length $n$ is bounded from above by a certain  
fixed 
linearly-bounded function in $n$; in other words, every computation 
tree of input size $n$  has hight of at most $O(n)$. (See \cite{TYL04} 
for more discussions on the topics of this strong definition of 
running time.) 

{\em One-way finite (state) automata} are a special case of those linear-time 1TMs with the following restrictions: (i) a tape head always moves from left to right without stopping, (ii) it halts just after scanning the right endmarker $\dollar$, and (iii) the tape is read-only. Let $\reg$, $\cfl$, and $\dcfl$ denote respectively the families of {\em regular} languages, of {\em context-free} languages, and of {\em deterministic context-free} languages.   

In this paper, all 1PTMs use only {\em rational} transition probabilities. We say that a 1PTM $M$ {\em recognizes} a language $L$ if, for every input string $x$, (i) if $x\in L$, then $M$ accepts $x$ with probability $>1/2$ and (ii) if $x\not\in L$, then $M$ rejects $x$ with probability 
$\geq 1/2$, where ``probability'' is taken over all inner coin tosses of $M$ on the input $x$. Moreover, $M$ has {\em bounded error} if there exists a constant (also called an {\em error bound}) 
$\varepsilon\in[0,1/2)$ such that, for every input string $x$, either $M$ accepts $x$ with probability at least $1-\varepsilon$ or $M$ rejects $x$ with probability at least $1-\varepsilon$; otherwise, $M$ is said to have {\em unbounded error}. Following \cite{TYL04}, we denote by  
$\onedlin$ ($\onebplin$, $\oneplin$, resp.) the family of all languages that are recognized by 1DTMs (1PTMs with bounded error, 1PTMs with unbounded error, resp.) in linear time. The family $\onecequallin$ (pronounced ``one C equal LIN'') \cite{TYL04} is the collection of all languages $L$ for which there exist a linear-time 1PTM $M$ that satisfies the following condition: for every input $x$, $x\in L$ iff $M$ accepts $x$ with probability exactly $1/2$. 

To feed a piece of supplemental information together with an 
input string to 1TMs, 
we use a ``track'' notation of \cite{TYL04}. For two symbols $\sigma\in\Sigma$ and $\tau\in\Gamma$, the notation $\track{\sigma}{\tau}$ expresses a new symbol made from $\sigma$ and $\tau$. For a 1TM equipped with an input/work tape, this symbol $\track{\sigma}{\tau}$ is written in a single cell, which consists of two tracks, whose upper track contains $\sigma$ and the lower track contains 
$\tau$. For two strings $x$ and $y$ of the same length $n$, 
$\track{x}{y}$ succinctly denotes the string $\track{x_1}{y_1}\track{x_2}{y_2}\cdots\track{x_n}{y_n}$ of length $n$, provided that $x=x_1x_2\cdots x_n$ and $y=y_1y_2\cdots y_n$. Notice that a tape head of the 1TM scans two symbols $\sigma$ and $\tau$ in the symbol $\track{\sigma}{\tau}$ simultaneously as a single symbol.  This track notation can be further extended to the case where $|x|\neq|y|$. If $|x|<|y|$ and $y=y_1y_2$ with $|x|=|y_1|$, then $\track{x}{y}$ denotes $\track{x}{y_1}\track{\#^{|x|-|y|}}{y_2}$; if $|x|>|y|$ and $x=x_1x_2$ with $|y|=|x_1|$, then $\track{x}{y}$ denotes $\track{x_1}{y}\track{x_2}{\#^{|x|-|y|}}$, where $\#$ is a special 
symbol representing a ``blank.'' 

For our later use, we also give a description of probabilistic finite automata. 
Here, we assume that all vectors are always expressed as {\em row vectors}. The notation $M^{T}$ for a matrix $M$ denotes the {\em transposed matrix} of $M$. A matrix is called {\em stochastic} if 
every row of it sums up to exactly $1$.  
A {\em one-way (rational) probabilistic finite automaton} (or 1pfa, in short) $M$ is a quintuple  $(Q,\Sigma,\nu_{ini},\{M_{\sigma}\}_{\sigma\in\check{\Sigma}},F)$, 
where $Q$ is a finite set of inner states, $\Sigma$ is an alphabet, 
$\nu_{ini}$ is an {\em initial state vector} with rational entries, each $M_{\sigma}$ is a $|Q|\times|Q|$ stochastic matrix  with rational entries, and $F$ ($\subseteq Q$) is a set of {\em final states}. 
The set $F$ induces a vector $\xi_{F}$ defined as follows: for each 
state $q\in Q$, the $q$-entry of $\xi_{F}$ has value $1$ if $q\in F$, and $0$ otherwise. Without loss of generality, we can assume that 
$\nu_{ini}$ has always value $1$ in its $q_{0}$-entry and $0$ in all the other entries. 
For each sequence $x=\sigma_1\sigma_2\cdots\sigma_n$ in 
$\check{\Sigma}^n$, $M_{x}$ is shorthand for $M_{\sigma_1}M_{\sigma_2}\cdots M_{\sigma_n}$. The {\em acceptance} ({\em rejection}, resp.) {\em probability} of $M$ on input $x\in\Sigma^*$ is defined as $p_{acc}(x) = \nu_{ini}M_{\cent x\dollar}\xi_{F}^{T}$ ($p_{rej}(x) = 1- p_{acc}(x)$, resp.).  

\section{Deterministic Computation with Advice}

We formally define the notion of {\em (deterministic) advice}\footnote{In 
the literature, there are at least two different formulations of ``advice'' for the model of one-tape machines: Damm and Holzer's \cite{DH95} and Tadaki, Yamakami, and Lin's \cite{TYL04}. These definitions are, however, computationally equivalent for, \eg polynomial time-bounded computations. Theorem \ref{REG/n-character} gives an implicit justification 
for the choice of our advice model.} and describe how to use it on one-tape linear-time 
Turing machines and finite automata. An {\em advice function} $h$ is a function mapping $\nat$ to $\Gamma^*$, where $\Gamma$ is a certain alphabet (which is particularly 
referred to as an {\em advice alphabet}). 
The advised language family $\onedlin/lin$ ($\reg/n$, resp.) is defined in \cite{TYL04} as the collection of all languages 
$S$ over alphabets $\Sigma$ such that there exist an advice alphabet $\Gamma$, an advice function $h:\nat\rightarrow\Gamma^*$, 
and a linear-time 1DTM (a 1dfa, resp.) $M$ satisfying the following conditions: (i) there are two constants $c,d>0$ such that, for every length $n\in\nat$, $|h(n)|\leq cn+d$ ($|h(n)|=n$, resp.) 
and (ii) for every string $x\in\Sigma^*$, $x\in S$ iff $M$ accepts $\track{x}{h(|x|)}$ (notationally, $M(\track{x}{h(|x|)})= S(x)$). 
Note that $\reg/n$ contains non-regular languages, for instance, the language $L_{eq}=\{0^n1^n\mid n\in\nat\}$.  
Surprisingly, the above two advised families, $\onedlin/lin$ and 
$\reg/n$,  coincide [10, Proposition 4.11]. 

\begin{lemma}\label{onedlin-vs-reg}
{\rm \cite{TYL04}}\hs{1}
$\onedlin/lin = \reg/n$.
\end{lemma}

In a polynomial-time setting, Karp-Lipton's advice naturally induces  {\em non-uniform computations}. Similarly, languages in $\reg/n$ can be  characterized in a certain non-uniform fashion. Here, we present 
a simple form of such non-uniform characterization of every 
language in $\reg/n$. 

\begin{theorem}\label{REG/n-character}
For any language $S$ over an alphabet $\Sigma$, the following two statements are equivalent. Let $\Delta = \{(x,n)\in \Sigma^*\times\nat\mid |x|\leq n\}$. 
\begin{enumerate}\vs{-1}
\item $S$ is in $\reg/n$.
\vs{-2}
\item There is an equivalence relation $\equiv_{S}$ over $\Delta$ such that 
\begin{enumerate}\vs{-2}
\item[(a)] the total number of equivalence classes in $\Delta/\equiv_{S}$ is finite, and
\vs{-1}
\item[(b)] for any length $n\in\nat$ and any two strings $x,y\in\Sigma^*$ with $|x|=|y|\leq n$, 
the following holds: $(x,n)\equiv_{S}(y,n)$ iff, for all $z$ with $|xz|=n$, $S(xz) = S(yz)$.
\end{enumerate}
\end{enumerate}
\end{theorem}

\begin{proof}
(1 $\Rightarrow$ 2) 
Assume that $S$ is a language in $\reg/n$ over an alphabet $\Sigma$. 
Take a 1dfa $M=(Q,\Lambda,q_0,\delta,Q_{acc},Q_{rej})$, an  advice alphabet $\Gamma$, and an advice function $h:\nat\rightarrow\Gamma^*$ satisfying that $|h(n)|=n$ and $S = \{x\in\Sigma^* \mid M \text{ accepts } \track{x}{h(|x|)}\}$, where $\Lambda = \{\track{\sigma}{\tau}\mid \sigma\in \Sigma,\tau\in\Gamma\}$. Without loss of generality, we can assume that $|Q_{acc}|=|Q_{rej}|=1$. This is possible by forcing $M$ to enter a unique accepting/rejecting state after scanning the right endmarker $\dollar$. 
For each length $n\in\nat$ and each string $x\in\Sigma^*$, 
we define $q_{n,x}$ as an inner state $q$ for which (i) if $|x|<n$, then $M$ enters $q$ just after reading $\cent\track{x}{w}$, where $w$ is  the first $|x|$ symbols of $h(n)$, and (ii) if $|x|=n$, then $M$ enters $q$ after reading $\cent\track{x}{h(n)}\dollar$. 
Now, let us define the desired relation $\equiv_{S}$ as follows: 
$(x,n)\equiv_{S} (y,m)$ iff $q_{n,x} = q_{m,y}$. 

Condition $(a)$ follows from the fact that $|\Delta/\equiv_{S}| = |\{q_{m,y}\mid m\in\nat,y\in\Sigma^*\}| \leq |Q|$. 
Next, we want to show Condition $(b)$. 
Consider two inputs $\track{xz}{wv}$ and $\track{yz}{wv}$ with $|x|=|y|=|w|$, $|xz|=n$, and $h(n)=wv$. Assume that $(x,n)\equiv_{S}(y,n)$. This means that, after reading $\cent\track{x}{w}$ as well as $\cent\track{y}{w}$, $M$ enters the same inner state, say, $q$. 
Hence, even if the two inputs $\track{xz}{wv}$ and $\track{yz}{wv}$ 
are different, $M$ behaves exactly in the same way during reading $\track{z}{v}\dollar$. As a consequence, for any $x,y,n$ with $|x|=|y|\leq n$, $q_{n,x} = q_{n,y}$ iff, for any $z$ with $|xz|=n$, there exists a 
halting state $q'$ in $Q_{acc}\cup Q_{rej}$ 
satisfying that $q_{n,xz} = q_{n,yz} = q'$. 
Condition $(b)$ then follows immediately.

(2 $\Rightarrow$ 1) 
To make our proof simple, we ignore the empty string and consider only the set $A-\{\lambda\}$. Given a language $S$, an equivalence relation $\equiv_{S}$ is assumed to satisfy Conditions $(a)$-$(b)$. We aim at showing that $S$ belongs to $\reg/n$.  By Condition $(a)$, let $\Delta/{\equiv_{S}} = \{A_1,A_2,\ldots,A_d\}$ for a certain constant $d\in\nat^{+}$. 
For these equivalence classes $A_q$ 
in $\Delta/\equiv_{S}$,  the following two properties hold.

\begin{claim}\label{1dfa-property}
Let $n\in\nat^{+}$, $\sigma\in\Sigma$, $q\in[d]$, and $x,y\in\Sigma^*$ 
with $|x|=|y|<n$.
\begin{enumerate}\vs{-1}
\item If $(x,n),(y,n)\in A_{q}$, then there exists a unique index $q'\in [d]$ such that $(x\sigma,n),(y\sigma,n)\in A_{q'}$. 
\vs{-2}
\item There exist two different indices, say,  $q^{(n)}_{acc}$ and $q^{(n)}_{rej}$ in $[d]$ satisfying that $\{(z,n)\mid |z|=n, z\in S\}\subseteq A_{q^{(n)}_{acc}}$ and $\{(z,n)\mid |z|=n, z\not\in S\}\subseteq A_{q^{(n)}_{rej}}$. 
\end{enumerate}
\end{claim}

\begin{proof}
(1) Let $x$ and $y$ satisfy $|x|=|y|<n$. We want to show that $(x,n)\equiv_{S}(y,n)$ implies $(x\sigma,n)\equiv_{S}(y\sigma,n)$. Note that the ``uniqueness'' requirement follows from the fact that $A_{q_1}\cap A_{q_2}=\setempty$ for any distinct indices 
$q_1,q_2\in[d]$. 
Toward a contradiction, we assume that $(x,n)\equiv_{S}(y,n)$ but $(x\sigma,n)\not\equiv_{S}(y\sigma,n)$. Take any string $z$ with $|z|=n-|x|-1$. Notice that such a $z$ exists since $|x|\leq n-1$. 
First, consider the case where $x\sigma z \in S$. 
Since $(x\sigma,n)\not\equiv_{S}(y\sigma,n)$, by Condition (b), we have $y\sigma z\not\in S$, which further implies $(x,n)\not\equiv_{S}(y,n)$. 
This is obviously a contradiction against our assumption. We therefore conclude that $(x\sigma,n)\equiv_{S}(y\sigma,n)$. The case where $x\sigma z\not\in S$ is similar. 

(2) Assume that $|x|=|y|=n$. As a special case of Condition $(b)$, it follows that 
(*)  $(x,n)\equiv_{S}(y,n)$ iff $S(x)=S(y)$. 
In particular, if $x,y\in S$, then  $(x,n)\equiv_{S}(y,n)$. Therefore, there exists an index $e\in[d]$ such that $\{(x,n)\mid |x|=n,x\in S\}\subseteq A_{e}$. We write this $e$ as $q_{acc}^{(n)}$. Likewise, 
we have $\{(x,n)\mid |x|=n,x\not\in S\}\subseteq A_{e'}$ for another index  $e'\in[d]$. By Statement (*), $e\neq e'$ follows immediately. Write this $e'$ as $q_{rej}^{(n)}$, and we then obtain the claim. 
\end{proof}

Let us return to the proof of the theorem. For any length 
$n\in\nat^{+}$ and any index $i\in[n]$, we define a series of 
finite functions $h_{n,i}:[d]\times \Sigma\rightarrow [d]\cup[d]^{3}$ 
as follows.
\begin{enumerate}
\item[(i)] Let $h_{n,1}(q,\sigma) = q'$ if $(\sigma,n)\in A_{q'}$. 
\vs{-2}
\item[(ii)] For any index $i$ with $1 < i < n$, let $h_{n,i}(q,\sigma) =q'$ if  there exists a string $x$ with $|x|=i-1$ such that $(x,n)\in A_{q}$ and $(x\sigma,n)\in A_{q'}$. 
\vs{-2}
\item[(iii)] Let $h_{n,n}(q,\sigma) = (q',q^{(n)}_{acc},q^{(n)}_{rej})$ if there exists a string $x$ with $|x|=n-1$ such that $(x,n)\in A_{q}$ and $(x\sigma,n)\in A_{q'}$.
\end{enumerate}
Hereafter, we treat each $h_{n,i}$ as a new symbol and define  $\Gamma=\{h_{n,i}\mid n\geq 1, i\in[n]\}$. It is important to note that $\Gamma$ is a finite set. Our advice string $h_n$ of length $n$ is defined to be $h_{n,1}h_{n,2}\cdots h_{n,n}$. 

At this point, we need to show that $h_{n,i}$ is indeed a {\em function}. Consider the case where $1<i<n$. Assuming that $h_{n,i}(q,\sigma)=q_1$ and $h_{n,i}(q,\sigma)=q_2$, we take two strings $x_1$ and $x_2$ of length $<n$ satisfying that $(x_1,n),(x_2,n)\in A_{q}$, $(x_1\sigma,n)\in A_{q_1}$, and $(x_2\sigma,n)\in A_{q_2}$. The uniqueness condition of Claim~\ref{1dfa-property}(1) then yields the desired equality $q_1=q_2$. The  other cases for $h_{n,1}$ and $h_{n,n}$ can be similarly treated. 

Now, let us define a finite automaton $M$ with its transition function $\delta$ as follows. We prepare four new inner states $q_{acc}$, $q_{rej}$, $q'_{acc}$, and $q'_{rej}$, which do not appear in $Q$. Let $\delta(q_0,\cent)=q_0$ and, moreover, let $\delta(q,\track{\sigma}{h_{n,i}}) = h_{n,i}(q,\sigma)$ for every index $i\in[n-1]$. When $h_{n,n}(q,\sigma) = (q',q^{(n)}_{acc},q^{(n)}_{rej})$,  let $\delta(q,\track{\sigma}{h_{n,n}}) = q'_{acc}$ if $q'=q^{(n)}_{acc}$; let $\delta(q,\track{\sigma}{h_{n,n}}) = q'_{rej}$ if $q'=q^{(n)}_{rej}$. Finally, let $\delta(q'_{acc},\dollar) = q_{acc}$ and 
$\delta(q'_{rej},\dollar) = q_{rej}$. 

In this end, we want to show that $x\in S$ iff $M$ accepts $\track{x}{h_n}$. Let $x=\sigma_1\sigma_2\cdots \sigma_n$ and assume that $(\lambda,n)\in A_{q_0}$, $(\sigma_1,n)\in A_{q_1}$,  $(\sigma_1\sigma_2,n)\in A_{q_2}$, $\ldots$, $(x,n)\in A_{q_n}$. First, let us  consider the case where $x\in S$. {}From now, we intend to prove by induction that 
\begin{quote}
\hs{15} (**) \hs{10} $q_i = \hat{\delta}(q_0,\cent\track{\sigma_1\cdots\sigma_i}{h_{n,1}\cdots h_{n,i}})$ for every $i\in[0,n-1]_{\integer}$,
\end{quote}
where $\hat{\delta}$ is the extended transition function induced from $\delta$. The basis case $i=0$ holds since 
$\hat{\delta}(q_0,\cent)=q_0$. {}From the induction hypothesis on 
$i<n-1$, it follows that
\begin{eqnarray*}
\hat{\delta}(q_0,\cent\track{\sigma_1\cdots\sigma_{i+1}}{h_{n,1}\cdots h_{n,i+1}})
&=& \delta(\hat{\delta}(q_0,\cent\track{\sigma_1\cdots\sigma_i}{h_{n,1}\cdots h_{n,i}}), \track{\sigma_{i+1}}{h_{n,i+1}}) \\
&=& h_{n,i+1}(\hat{\delta}(q_0,\cent\track{\sigma_1\cdots\sigma_i}{h_{n,1}\cdots h_{n,i}}), \sigma_{i+1}) \\
&=& h_{n,i+1}(q_i,\sigma_{i+1}) 
\;\;=\;\; q_{i+1}. 
\end{eqnarray*}
Thus, Statement (**) holds. In particular, we have 
$q_{n-1} = \hat{\delta}(q_0,\cent\track{\sigma_1\cdots\sigma_{n-1}}{h_{n,1}\cdots h_{n,{n-1}}})$. Note that, by the definition of $h_{n,n}$,  $h_{n,n}(q_{n-1},\sigma_n) = (q_n,q_{acc}^{(n)},q_{rej}^{(n)})$. 
Since $x\in S$, Claim~\ref{1dfa-property}(2) implies $A_{q_n} = A_{q^{(n)}_{acc}}$, from which we obtain  
\[
\hat{\delta}(q_0,\cent\track{x}{h_n}) 
= \delta(\hat{\delta}(q_0,\cent\track{\sigma_1\cdots\sigma_{n-1}}{h_{n,1}\cdots h_{n,n-1}}), \track{\sigma_{n}}{h_{n,n}})
= \delta(q_{n-1},\track{\sigma_n}{h_{n,n}}) = q'_{acc} 
\]
and thus $\hat{\delta}(q_0,\cent\track{x}{h_n}\dollar) = q_{acc}$. 
This means that $M$ accepts $\track{x}{h_n}$. 
 The other case $x\not\in S$ is similar to the previous case, since the only difference is the final step of the above argument. This completes the proof of Theorem \ref{REG/n-character}.
\end{proof}

\section{Probabilistic Computation with Advice}

Probabilistic computation has been a useful tool for designing many 
practical algorithms. We shall move our attention to linear-time 1PTMs, 
supplemented with deterministic advice. In a similar fashion to $\onedlin/lin$, we 
define two families  $\oneplin/lin$ and $\onecequallin/lin$ by simply modifying the definition of $\onedlin/lin$ using 1PTMs in lieu of 1DTMs. 

\sloppy Earlier, Tadaki \etalc~\cite{TYL04} showed that $\cfl\cap \onecequallin\nsubseteq \reg/n$. {}From their result, since $\reg/n\subseteq \onecequallin/lin$, it immediately follows that $\reg/n$ is properly included in $\onecequallin/lin$ (as well as $\co\onecequallin/lin$).

\begin{proposition}\label{reg-vs-cequallin}
$\reg/n\subsetneqq \onecequallin/lin\cap \co\onecequallin/lin$.
\end{proposition}

The above proposition indicates that, even in the presence of advice, probabilistic computation is much more powerful than deterministic computation. Naturally, we can question how powerful the families $\onecequallin$ 
and $\oneplin$ are when deterministic advice is allowed.  
We shall provide several answers to this question. 
Our first answer is the following.  

\begin{theorem}\label{cequallin-vs-plin}
$\onecequallin/lin\neq \co\onecequallin/lin \neq \oneplin/lin$.
\end{theorem}

To prove this theorem, we need a key lemma, which gives a new criterion that every language in $\onecequallin/lin$ must satisfy. Historically, in early 1970s, Dieu \cite{Die71} showed the following criterion (in our terminology): if $L$ is in $\onecequallin$, then there exists a number $m\in\nat^{+}$ such that, for any $u,y,v\in\Sigma^*$, $\{uy^iv\mid i\in[0,m-1]_{\integer}\}\subseteq L$ implies 
$\{uy^iv\mid i\in\nat\}\subseteq L$. Unfortunately, his criterion cannot be extended to our advised language family $\onecequallin/lin$, because advice strings may change as input size increases. 
To prove our theorem, we should seek another criterion for $\onecequallin/lin$. The next lemma provides one such criterion. 

\begin{lemma}\label{cequallin-lemma}
Let $A\in\onecequallin/lin$ be any language over an alphabet $\Sigma$. There exists a positive integer $m$ that satisfies the following statement. Let $n,\ell\in\nat$ and $z\in\Sigma^*$ satisfy that $n\geq 1$, $\ell\leq n-1$, and $|z|=\ell$. Let $A_{n,z}=\{w\in\Sigma^{n-\ell}\mid wz\in A\}$. There exists a subset $S\subseteq A_{n,z}$ with $|S|\leq m$ such that, for each string  $y\in\Sigma^{\ell}$, if $\{wy\mid w\in S\}\subseteq A$ then $\{xy\mid x\in A_{n,z}\}\subseteq A$.  
\end{lemma}

Before proving  Lemma \ref{cequallin-lemma}, we shall present the proof 
of Theorem \ref{cequallin-vs-plin}. This proof exemplifies usefulness of the criterion given in the lemma.

\begin{proofof}{Theorem \ref{cequallin-vs-plin}}
Let $\Sigma=\{0,1\}$ and consider the language $Dup=\{ww\mid w\in\Sigma^*\}$ ({\em duplicated strings}). Later, in the proof of Proposition \ref{reg-cequal-cfl}, we shall prove that $Dup\in\onecequallin/lin$. It is therefore sufficient to show below that  
$\overline{Dup}\not\in \onecequallin/lin$. 

For simplicity, we write $A$ for $\overline{Dup}$. Now, we want to show that $A\not\in\onecequallin/lin$. To lead to a contradiction, we assume  that $A\in\onecequallin/lin$. By Lemma \ref{cequallin-lemma}, there is a positive integer $m$ that satisfies the conclusion of the lemma. Now, we choose the minimal even integer $n$ such that $2^{n/2}-1>m+1$. 
Define $z=1^{n/2}$. Clearly, the set $A_{n,z}$ ($= \{w\in\Sigma^{n/2}\mid wz\in A\}$) satisfies that $A_{n,z} = \Sigma^{n/2} \setminus \{z\}$. 
There exists a set $S\subseteq A_{n,z}$ with $|S|\leq m$ that satisfies the lemma. Take any string $y$ in $\Sigma^{n/2}\setminus (S\cup\{z\})$. Note that $y\in A_{n,z}$ because  $A_{n,z} = \Sigma^{n/2}\setminus \{z\}$ and $y\neq z$. 
Since $wy\in A$ for any string $w\in S$, the lemma implies that $yy\in A$. This is a contradiction against the fact that $yy\not\in \overline{Dup} = A$.  
Thus, we conclude that $A\not\in\onecequallin/lin$. As a result, we obtain the desired separation: $\onecequallin/lin \neq \co\onecequallin/lin$. 

Next, we assume that $\co\onecequallin/lin = \oneplin/lin$. Similar to $\oneplin$, the advised family $\oneplin/lin$ is closed under complementation; that is, $\oneplin/lin = \co\oneplin/lin$. This closure property implies that $\onecequallin/lin$ is also closed under complementation; however, this contradicts the above separation result between $\onecequallin/lin$ and $\co\onecequallin/lin$.  It therefore 
follows that $\co\onecequallin/lin \neq \oneplin/lin$, as requested.  
\end{proofof}

We shall present the proof of Lemma \ref{cequallin-lemma}. For this proof, 
we shall use the following characterization of $\onecequallin/lin$, described in terms of 1pfa's and advice functions. This characterization is in essence analogous to Lemma \ref{onedlin-vs-reg}. Since it can be shown by an argument similar to the proof of Claim \ref{reduce-1PTM-1pfa} (which uses a notion of ``folding machine'' in \cite{TYL04}), we omit its proof for readability. 

\begin{lemma}\label{1C=LIN/lin-character}
For any language $A$ over an alphabet $\Sigma$, $A\in\onecequallin/lin$ iff there exist a 1pfa  $M$ and an advice function $h$ that satisfy the following: for every string $x\in\Sigma^*$, $x\in A$ iff 
$M$ accepts $\track{x}{h(|x|)}$ with probability exactly $1/2$. 
\end{lemma}

The following proof of Lemma \ref{cequallin-lemma} exploits a fundamental property of a stochastic matrix, which states that each of its rows sums up to $1$.  

\begin{proofof}{Lemma \ref{cequallin-lemma}}
Let $A\in\onecequallin/lin$ be any language over an alphabet $\Sigma$. 
Lemma \ref{1C=LIN/lin-character} guarantees the existence of a 1pfa  $M=(Q,\Lambda,\nu_{ini},\{M_{\sigma}\}_{\sigma\in\check{\Lambda}},F)$, an advice alphabet $\Gamma$, and an advice function $h:\nat\rightarrow\Gamma^*$ such that, 
for every string $x\in\Sigma^*$, $x\in A$ iff  $\prob_{M}[M(\track{x}{h(|x|)}) = 1] = 1/2$, where $\Lambda = \{\track{\sigma}{\tau}\mid \sigma\in\Sigma,\tau\in\Gamma\}$.
Recall from Section \ref{sec:basic} our assumption on $\nu_{ini}$ and $\xi_{F}$. 

We set $m= |Q|$ and choose $n,\ell\in\nat$ and $z\in\Sigma^{\ell}$ 
arbitrarily. Since the lemma trivially holds for $\ell=0$,  let us assume 
that $\ell>0$. 
For simplicity, let $h(n) =rs$ for two strings $r\in \Sigma^{n-\ell}$ and $s\in \Sigma^{\ell}$. 
Now, we focus our attention to the set $A_{n,z}=\{w\in\Sigma^{n-\ell}\mid wz\in A\}$. Notice that the lemma is also true by setting $S=A_{n,z}$ 
when $|A_{n,z}|\leq|Q|$. It thus suffices to consider the case where $|A_{n,z}|>|Q|$. 

Henceforth, we write $\tilde{w} = \track{w}{r}$, $\tilde{x} = \track{x}{r}$, $\tilde{z}=\track{z}{s}$, and $\tilde{y}=\track{y}{s}$ for notational convenience. Note that the acceptance probability $p_{acc}(wz) =_{def} \nu_{ini}M_{\cent\tilde{w}}M_{\tilde{z}\dollar}\xi_{F}^{T}$ (note that this notation suppresses the advice string $rs$) equals $1/2$ for all strings $wz$ in $A$. 
Choose a maximal subset $S'$ of 
linearly-independent vectors (which form a set of {\em basis vectors}) in the set $V=\{\nu_{ini}M_{\cent\tilde{w}} \mid w\in A_{n,z}\}$.  Clearly, since each vector $\nu_{ini}M_{\cent\tilde{w}}$ has dimension $|Q|$, there are at most $|Q|$ linearly-independent vectors in $V$. This implies that $|S'|\leq |Q|=m$. The desired set  $S$ is now defined as $S=\{w\in A_{n,z} \mid \nu_{ini}M_{\cent\tilde{w}}\in S'\}$. 
Note that any vector $\nu_{ini}M_{\cent\tilde{x}}$ in $V-S'$ can be written as a linear combination of basis vectors in $S'$: 
\[
(*) \hs{10} \nu_{ini}M_{\cent\tilde{x}} = \sum_{w\in S} \alpha_{w}\cdot \nu_{ini}M_{\cent\tilde{w}},
\]
where $\{\alpha_{w}\}_{w\in S}$ is a set of appropriate real numbers. 
As a key claim, we show the following statement.

\begin{claim}\label{sum-up-to-one}
$\sum_{w\in S}\alpha_{w} = 1$.
\end{claim}

\begin{proof}
For our convenience, let $\nu_{ini}M_{\cent\tilde{x}} = (a_{x,i})_{i\in[m]}$ and  let $\nu_{ini}M_{\cent\tilde{w}} = (a_{w,i})_{i\in[m]}$ for each element $w\in S$. 
{}From Equation (*), we obtain 
$
(a_{x,i})_{i\in[m]} = 
\left(\sum_{w\in S}\alpha_{w}a_{w,i}\right)_{i\in[m]}. 
$
For every string $w\in S$, since $M_{\cent\tilde{w}}$ is a stochastic matrix, it holds that  $\sum_{i=1}^{m}a_{w,i}=1$. Likewise, we have $\sum_{i=1}^{m}a_{x,i}=1$. {}From these equations, we obtain  $\sum_{i=1}^{m}\left(\sum_{w\in S}\alpha_{w}a_{w,i}\right) = 1$, which further implies  
\[
1  = \sum_{i=1}^{m}\left(\sum_{w\in S}\alpha_{w}a_{w,i}\right) 
= \sum_{w\in S}\alpha_{w}\left(\sum_{i=1}^{m}a_{w,i}\right) =  \sum_{w\in S}\alpha_{w}.
\]
This finishes the proof of the claim.
\end{proof}

At last, we shall show the desired property of $S$: for every string  $y\in\Sigma^{\ell}$, $\{wy\mid w\in S\}\subseteq A$ implies $\{xy\mid x\in A_{n,z}\}\subseteq A$. To show this property, let $y$ be any string in $\Sigma^{\ell}$ and assume that $wy\in A$ for all strings $w$ in $S$. {}From Claim~\ref{sum-up-to-one}, since $p_{acc}(wy)=1/2$, it follows by Equation (*) that, for each string $x\in A_{n,z}$, 
\begin{eqnarray*}
p_{acc}(xy) &=& \nu_{ini}M_{\cent\tilde{x}}M_{\tilde{y}\dollar}\xi_{F}^{T} 
\;\;=\;\; \left(\sum_{w\in S} \alpha_{w}\cdot \nu_{ini}M_{\cent\tilde{w}}\right)M_{\tilde{y}\dollar}\xi_{F}^{T} \\
&=& \sum_{w\in S} \alpha_{w} \left(\nu_{ini}M_{\cent\tilde{w}} M_{\tilde{y}\dollar}\xi_{F}^{T} \right) 
\;\;=\;\; \sum_{w\in S} \alpha_{w} \cdot p_{acc}(wy) \\
&=& \frac{1}{2} \sum_{w\in S}\alpha_{w} 
\;\;=\;\; \frac{1}{2}. 
\end{eqnarray*}
Hence, we obtain $xy\in A$. Because $x$ is arbitrary in $A_{n,z}$, 
we finally obtain the desired property and therefore the lemma. 
\end{proofof}

The advised language family $\cfl/n$ was introduced in \cite{Yam08a}, analogous to $\reg/n$, using {\em one-way nondeterministic pushdown automata}\footnote{Roughly speaking, a 1npda is a one-way nondeterministic automaton equipped with a {\em stack}, which is an additional read/write tape 
whose access is regulated by the first-in, last-out policy.} 
(or 1npda's, in short) together with deterministic advice 
whose length equals input size. 
A class separation between $\cfl/n$ and $\onecequallin/lin$ is also 
possible. 

\begin{proposition}
\label{cequallin-vs-cfln}
$\onecequallin\nsubseteq \cfl/n$. Thus, $\onecequallin/lin \nsubseteq \cfl/n$. 
\end{proposition} 

\begin{proof}
Let  $\Sigma_6=\{a_1,a_2,\ldots,a_6,\#\}$ be our alphabet. Consider the language $Equal_{6}$ composed of all strings $w$ over $\Sigma_6$  such that  $\#_{a}(w)=\#_{a'}(w)$ for any pair $a,a'\in\Sigma_6-\{\#\}$. Since $Equal_{6}$ sits outside of $\cfl/n$ \cite{Yam08a}, it is enough to show that $Equal_{6}$ belongs to  $\onecequallin$. For any two fixed indices $i,j\in[1,6]_{\integer}$, we denote by $L_{i,j}$ the language $\{w\in\Sigma_6^*\mid \#_{a_i}(w)=\#_{a_j}(w)\}$. It is not difficult to show that each language $L_{i,j}$ belongs to $\onecequallin$. Moreover, note that $Equal_{6} = \bigcap_{i=2}^{6}L_{1,i}$. A crucial point is 
that $\onecequallin$ is closed under intersection  \cite{TYL04}.  
This closure property implies that $Equal_{6}$ belongs to $\onecequallin$. 
\end{proof}

As an immediate consequence of Proposition \ref{cequallin-vs-cfln}, we obtain $\oneplin/lin\nsubseteq\cfl/n$. In the following, we shall prove the other direction: $\cfl/n\nsubseteq \oneplin/lin$.  Notice that the non-advice separation $\cfl\nsubseteq\oneplin$ of  Nasu and Honda \cite{NH71} (see also [10, Proposition 6.7]) does not imply our desired 
separation.

\begin{theorem}\label{cfl-vs-oneplin}
$\cfl\nsubseteq\oneplin/lin$. Thus, $\cfl/n\nsubseteq \oneplin/lin$.
\end{theorem}

This theorem follows from the next lemma, 
which gives a new criterion for languages in $\oneplin/lin$. 
This lemma sharply contrasts with Lemma \ref{cequallin-lemma}.

\begin{lemma}\label{oneplin-criteria}
Let $A\in\oneplin/lin$ over an alphabet $\Sigma$. There exists a positive constant $m$ that satisfies the following statement. Let $n,\ell\in\nat$ be arbitrary numbers with $n\geq1$ and $\ell\leq n-1$. There exists a set $S=\{w_1,w_2,\ldots,w_m\}\subseteq \Sigma^{n-\ell}$ with $|S|=m$ for which the following implication holds: for any set $T\subseteq \Sigma^{\ell}$, if   $|\{A(w_1y)A(w_2y)\cdots A(w_my)\mid y\in T\}|\geq 2^{m}$, then it follows that, for any string $x\in\Sigma^{n-\ell}$, there exists a pair $y,y'\in T$ such that $xy\in A$ and $xy'\not\in A$. 
\end{lemma}

{}From this lemma, Theorem \ref{cfl-vs-oneplin} easily follows. 
In the following proof of the theorem, the notation $\odot$ denotes 
the {\em bitwise  binary inner product}. 

\begin{proofof}{Theorem \ref{cfl-vs-oneplin}}
Let $\Sigma=\{0,1\}$ for simplicity and consider the language $IP_{*}=\{axy\mid a\in\{\lambda,0,1\}, x,y\in\Sigma^*, |x|=|y|,x^{R}\odot y\equiv 0\;(\mathrm{mod}\;2)\}$.
Since $IP_{*}$ is in $\cfl$ \cite{Yam09a}, we want to show that $IP_{*}\not\in \oneplin/lin$.

Now, assuming that $IP_{*}\in\oneplin/lin$, we take a positive constant 
$m$ that satisfies Lemma \ref{oneplin-criteria}. Choose a sufficiently large even number $n$ and let $\ell=n/2$. There exists a subset $S=\{w_1,w_2,\ldots,w_{m}\}$ of $\Sigma^{\ell}$ with $m$ distinct elements that satisfy the lemma. 
To each binary sequence $r=(r_1,r_2,\ldots,r_m)$ (seen as a string) in $\Sigma^m$,  
we assign a certain string $y_{r}\in\Sigma^{n-\ell}$ satisfying that $w_i^{R}\odot y_{r}\equiv r_i \;(\mathrm{mod}\;2)$ for all indices  
$i\in[m]$. 
Note that $y_r \neq y_{r'}$ for any distinct pair $r,r'$. By collecting all such $y_r$'s, we define $T=\{y_{r}\mid r\in\Sigma^m\}$. Clearly, 
$|T|=2^m$ and therefore $|\{IP_{*}(w_1y)\cdots IP_{*}(w_my)\mid y\in T\}| = |T|\geq 2^{m}$.  
Since $n$ is sufficiently 
larger than $2^m$, there is a string $x\in\Sigma^{\ell}$ satisfying $x^{R}\odot y_r\equiv0\;(\mathrm{mod}\;2)$ for every sequence $r$ in $\Sigma^m$. In other words, $xy\in IP_{*}$ holds for any $y\in T$. For this $x$, the lemma yields a pair $y,y'\in T$ for which $xy\in IP_{*}$ and $xy'\not\in IP_{*}$. This is a contradiction against the choice of $x$. Hence, we conclude that $IP_{*}\not\in\oneplin/lin$.  
\end{proofof}

Finally, we shall give the proof of Lemma \ref{oneplin-criteria}. This proof relies on, similar to Lemma \ref{1C=LIN/lin-character}, a new 1pfa-characterization of $\oneplin/lin$.  

\begin{lemma}\label{character-PLIN}
For any language $A$ over an alphabet $\Sigma$, $A\in\oneplin/lin$ iff there exist a 1pfa $M$ and an advice function $h$ satisfying the following: for every input $x\in\Sigma^*$, $x\in A$ iff $M$ accepts $\track{x}{h(n)}$ with probability more than $1/2$.
\end{lemma}

\begin{proofof}{Lemma \ref{oneplin-criteria}}
Let $A$ be any language in $\oneplin/lin$ over an alphabet $\Sigma$. For this $A$, we can take a 1pfa $M$, an advice alphabet 
$\Gamma$, and an advice function $h:\nat\rightarrow\Gamma^*$, as 
described in Lemma \ref{character-PLIN}. Let $M=(Q,\Lambda,\nu_{ini},\{M_{\sigma}\}_{\sigma\in\check{\Lambda}},F)$, where $\Lambda = \{\track{\sigma}{\tau}\mid \sigma\in \Sigma,\tau\in\Gamma\}$. 
To make our argument simple, we make the following extra 
requirement: 
the success probability of $M$ on any input string never becomes 
exactly $1/2$. This requirement can be easily met by an appropriate modification 
of the given 1pfa $M$ (see, \eg \cite{TYL04}). In the subsequent 
argument, we assume that $\prob_{M}[M(\track{x}{h(n)}) = A(x)]>1/2$ 
for every string $x$.

Choose $n,\ell\geq1$ arbitrarily with $\ell\leq n-1$. To follow the 
proof of Lemma \ref{cequallin-lemma}, we intend to use the same notations $\tilde{x} = \track{x}{r}$, $\tilde{y}=\track{y}{s}$, $\tilde{w} = \track{w}{r}$, and 
$\tilde{w}_i = \track{w_i}{r}$ for an advice string $h(n)=rs$. 
The only difference from the proof of Lemma \ref{cequallin-lemma} is that we do not need to take a fixed input string $z$ that forces $M$ to accept.  Here, we choose a maximal subset $S'$ of linearly-independent vectors in the set $\{\nu_{ini} M_{\cent\tilde{w}}\mid w\in \Sigma^{n-\ell}\}$ and we then define $S=\{w\in\Sigma^{n-\ell}\mid \nu_{ini} M_{\cent\tilde{w}}\in S'\}$. Set $m=|S|$ and let $S=\{w_1,w_2,\ldots,w_m\}$. Now, fix $x\in\Sigma^{n-\ell}$ 
arbitrarily. Claim \ref{sum-up-to-one} ensures the existence of 
a series $\{\alpha_{w_i}\}_{i\in[m]}$ of real numbers that satisfy 
(i) $\sum_{i=1}^{m}\alpha_{w_i}=1$ and (ii) 
$
\nu_{ini}M_{\cent\tilde{x}} = \sum_{i=1}^{m} \alpha_{w_i} \left( \nu_{ini}M_{\cent\tilde{w}_i} \right).
$

Next, let $T$ be an arbitrary subset of $\Sigma^{\ell}$ and assume that, for any binary series $r=(r_1,r_2,\ldots,r_m)\in\{0,1\}^m$, there exists a string $y_{r}\in T$ satisfying $A(w_1y_r)A(w_2y_r)\cdots A(w_my_r) = r$. For each string $y\in\Sigma^{\ell}$ and each index $i\in[m]$, let $p_{acc}(w_iy) = 1/2+\beta_{i,y}$ (as before, this notation suppresses the advice string $rs$) for a certain real number  $\beta_{i,y}\in[-1/2,1/2]-\{0\}$. For such $y$'s, we have 
\begin{eqnarray*}
p_{acc}(xy) &=& \sum_{i\in[m]} \alpha_{w_i} \left( \nu_{ini}M_{\cent\tilde{w}_i}M_{\tilde{y}\dollar}\xi_{F}^{T} \right) 
\;\;=\;\; \sum_{i\in[m]} \alpha_{w_i}p_{acc}(w_iy) \\
&=& \frac{1}{2} + \sum_{i\in[m]}\alpha_{w_i}\beta_{i,y} 
\;\;=\;\; \frac{1}{2} 
+ \sum_{i\in[m]}(-1)^{1-A(w_iy)}|\beta_{i,y}|\alpha_{w_i}.
\end{eqnarray*}

To complete the proof, we want to define a binary series $r=(r_1,\ldots,r_m)$ as follows: let $r_i=0$ (or equivalently 
$A(w_iy_r)=0$) if $\alpha_{w_i}<0$, and let $r_i=1$ (or $A(w_iy_r)=1$) 
if $\alpha_{w_i}\geq 0$. Let us consider the string $y_{r}$ associated with this $r$, and make $y_r$ as our desired string $y$. By the choice 
of $y$, we have 
\[
\sum_{i\in[m]}(-1)^{1-A(w_iy)}|\beta_{i,y}|\alpha_{w_i} = \sum_{i\in[m]}|\beta_{i,y}||\alpha_{w_i}| >0, 
\]
which implies $p_{acc}(xy)>1/2$. This means that $xy\in A$. Next, we define $r'$ as the bitwise negation of $r$ and write $y'$ for $y_{r'}$. Similar to the previous case, we have 
\[
\sum_{i\in[m]}(-1)^{1-A(w_iy')}|\beta_{i,y'}|\alpha_{w_i} = \sum_{i\in[m]}(-1)|\beta_{i,y'}||\alpha_{w_i}| <0,
\]
and thus we obtain $xy'\not\in A$. This completes the proof of the lemma. 
\end{proofof}

\section{Power of Randomized Advice}

We have so far discussed the roles of a single advice string given 
per each input length $n$.  Instead of giving such a deterministic 
string, we can provide  ``randomized'' advice strings, each of 
which is produced according to a certain fixed probability distribution. 
It turns out that such randomized advice often endows an
enormous power to its underlying machine's language recognition.

In this paper, {\em randomized advice} refers to a probability ensemble $\{D_n\}_{n\in\nat}$, in which each probability distribution $D_n$ is  defined over all advice strings of length $n$. Let $m\in\nat$, let $x$ be any input string, and let $D_m$ be any probability distribution over $\Gamma^m$, where $\Gamma$ is an advice alphabet. 
The notation $\track{x}{D_m}$ indicates a random variable that expresses $\track{x}{y}$, which is chosen randomly with probability $D_m(y)$ over all strings $y$ in $\Gamma^m$. Similarly, given a machine $M$, the notation $M(\track{x}{D_m})$ denotes a random variable expressing the outcome $M(\track{x}{y})$ of $M$ on an input $\track{x}{y}$, which is chosen randomly according to $D_m$. 

We use the notation $\onebplin/Rlin$ to denote the collection of all languages $L$ for which there exist a linear-time 1PTM $M$, an error bound 
$\varepsilon\in[0,1/2)$, a probability ensemble $\{D_n\}_{n\in\nat}$, and a linearly-bounded function $\ell:\nat\rightarrow\nat$ satisfying: for every input $x$, if $x\in L$ then $M$ accepts $\track{x}{D_m}$ with probability at least $1-\varepsilon$; otherwise, $M$ rejects $\track{x}{D_m}$ with probability at least $1-\varepsilon$,  where $m=\ell(n)$.  For brevity, we write 
$\prob_{M,D_{m}}[M(\track{x}{D_{m}})=L(x)]\geq 1-\varepsilon$. 
In a similar fashion, we may define $\onecequallin/Rlin$ and $\oneplin/Rlin$ with linear-size randomized advice; however, those two language families are so powerful that they can recognize {\em all}  languages.

\begin{proposition}
The advised language family $\onecequallin/Rlin$ 
as well as $\oneplin/Rlin$ consists of all languages.
\end{proposition}

\begin{proof}
Let $A$ be an arbitrary language over an alphabet $\Sigma$. 
Our goal is to show that $A$ belongs to $\onecequallin/Rlin$. 
To simplify our proof, we assume that $\lambda\not\in A$ and we hereafter consider only positive input lengths. 
For each length $n\in\nat^+$, we write $A_n$ for $A\cap\Sigma^n$ 
and $\overline{A}_n$ for $\Sigma^n\setminus A_n$. Moreover, we set $\Gamma=\Sigma\cup\{\#\}$, where $\#$ is a special symbol not in $\Sigma$. 

Let us define a probability distribution $D_n$ over $\Gamma^n$, which generates only strings in $\Sigma^n$ as well as the string $\#^n$ with positive probabilities. Let $y$ be any advice string in $\Gamma^n$. Whenever 
$y\not\in\Sigma^n\cup\{\#^n\}$, let $D_n(y)=0$. Henceforth, we are 
focused only on advice strings $y$ in $\Sigma^n\cup\{\#^n\}$. If $\overline{A}_n=\setempty$, then we set $D_n(y)=1$ if $y=\#^n$, and $D_n(y)=0$ 
for the other advice strings $y$. When $\overline{A}_n\neq\setempty$, 
we set $D_n(y)=1/|\overline{A}_n|$ if $y\in \overline{A}_n$, and $D_n(y)=0$ 
if $y\in A_n\cup\{\#^n\}$. 
Our 1pfa $M$ works as follows. Let $\track{x}{s}$ be an arbitrary input of length $n$. If $x=s$, then $M$ accepts $\track{x}{s}$ with certainty; otherwise, it accepts and rejects $\track{x}{s}$ with an equal probability. Note that, when $s=\#^n$, $\track{x}{s}$ is never accepted. 

In the case where $x\not\in A$, the acceptance probability of $M$ 
on $\track{x}{s}$ 
is at least $1/2+ D_n(x)/2$ ($=1/2+1/2|\overline{A}_n|$), which is clearly more than $1/2$. By contrast, when $x\in A$, the acceptance probability 
on $\track{x}{s}$ is exactly $1/2$ since $D_n(x)=0$. Therefore, $A$ belongs to $\onecequallin/Rlin$. Since $\onecequallin/Rlin\subseteq \oneplin/Rlin$, the proposition easily follows.
\end{proof}

We return to the advised language family 
$\onebplin/Rlin$. Earlier, Tadaki \etalc~\cite{TYL04} showed that $\onebplin$ coincides with $\reg$. 
Let us present a similar characterization of $\onebplin/Rlin$ using one-way finite automata. First, we introduce $\reg/Rn$---a natural   
extension of $\reg$ by supplying randomized advice.   
Formally, a language $A$ over an alphabet $\Sigma$ is in $\reg/Rn$ if there exist a 1dfa $M$, an error bound $\varepsilon\in[0,1/2)$, an advice alphabet $\Gamma$, and a probability ensemble  $\{D_n\}_{n\in\nat}$ over $\Gamma^*$ that satisfy the following condition: for any length 
$n\in\nat$ and any string $x\in\Sigma^n$, if $x\in A$ then $M$ accepts $\track{x}{D_n}$ with probability $\geq 1-\varepsilon$; otherwise, $M$ rejects $\track{x}{D_n}$ with probability $\geq 1-\varepsilon$. 
Likewise, we can define another advised language family $\cfl/Rn$ using 
1npda's (instead of 1dfa's) together with randomized advice. 
Obviously, $\reg/Rn\subseteq \cfl/Rn$. 
In comparison to Lemma \ref{onedlin-vs-reg}, we shall prove that the two families $\onebplin/Rlin$ and $\reg/Rn$ coincide. Notice that this result is not a direct consequence of the aforementioned equality $\onebplin = \reg$ of Tadaki \etalc~\cite{TYL04}; rather, it is from the fact that bounded-error probabilistic computation can be integrated into 
randomized advice. 

\begin{theorem}\label{1BPLIN-vs-REG}
$\onebplin/Rlin = \reg/Rn$.
\end{theorem}

\begin{proof}
Since every 1pfa can be simulated by a certain linear-time 1PTM, the inclusion $\reg/Rn\subseteq \onebplin/Rlin$ follows immediately. 
Hereafter, we pay our attention to the remaining inclusion. The 
following proof consists of three stages. 
Assume that $A$ is any language in $\onebplin/Rlin$, witnessed by a linear-time 1PTM $M$ and a probability ensemble $\{D_n\}_{n\in\nat}$ over the set $\Gamma^*$ of advice strings. To ease the notational complexity,  we use the same terminology given in \cite{TYL04}. In the first stage, we claim that this linear-time 1PTM $M$ can be replaced by a certain 1pfa  
$N$ even in the presence of randomized advice. 

\begin{claim}\label{reduce-1PTM-1pfa}
There exists a 1pfa $N$ and an error bound $\varepsilon\in[0,1/2)$ 
such that, for every length $n\in\nat$ and every string $x\in\Sigma^n$,  $\prob_{N,D_n}[N(\track{x}{D_n})=A(x)]\geq 1-\varepsilon$.
\end{claim}

\begin{proof}
{}From the length requirement for our randomized advice $\{D_n\}_{n\in\nat}$,  we can assume without loss of generality that, for every string $x\in\Sigma^*$, $\prob_{M,D_{m(|x|)}}[M(\track{x}{D_{m(|x|)}})=A(x)]\geq 1-\varepsilon$, where a length function $m$ satisfies that $n\leq m(n)\leq kn$ for every length $n\in\nat$.   
Now, we assume that $M$ takes an input of the form $\track{x\#^{|w|-n}}{w}$ with an advice string $w\in\Gamma^{m(n)}$. 

How can we simulate $M$'s moves along the tape cells indexed from 
$-2k(n-1)$ to $2k(n-1)-1$ using only its ``input area'' (\ie the tape region where the original input string $x$ is written, together 
with the two endmarkers)? 
As shown in [10, Section 4.2], it is possible to ``fold'' $M$'s tape content into its input area; namely, the tape content is partitioned into $4k$ blocks of size $n-1$, indexed from left to right by numbers 
between $-2k$ and $2k-1$. 
Each block's content is written in one of $4k$ tracks, indexed from 
top to bottom by the same numbers, of a new tape 
so that we use only $n$ tape cells to simulate $M$'s entire behavior on the input  $\track{x\#^{|w|-n}}{w}$. This gives rise to a so-called ``folding machine,'' which simulates $M$ using only its input area. 

Let $cont(x,w)$ denote a string obtained by folding the tape content $\track{x\#^{|w|-n}}{w}$ into its input area.  By deleting all symbols in $\Sigma$ from $cont(x,w)$, we obtain a new advice string, say, $w'$ of length $n$. For this $w'$, set  $D'_n(w')=D_{m(n)}(w)$. A new 1PTM $M'$ behaves as follows: on input $\track{x}{w'}$, first recover the string  $cont(x,w)$ in linear time and then simulate $M$'s folding machine 
using $cont(x,w)$ as its new input string. 

{}From the above argument, $M'$ takes new advice that is randomly distributed over strings of length equal to input size and 
moves its tape head between (and on) the two endmarkers. 
Apply [10, Lemma 6.5] to obtain a rational {\em one-way generalized probabilistic finite automaton} (or 1GPFA, in short) $\tilde{N}$ 
such that $p_{M'}(\track{x}{w})=p_{\tilde{N}}(\track{x}{w})$ for any pair $x,w$. We then simulate this 1GPFA by another rational 1pfa $N$ with preserving the same acceptance/rejection probability. This proves the claim. 
\end{proof}

Now, we assume that $N$ is a 1pfa working with the randomized advice $\{D'_n\}_{n\in\nat}$.
Let $N=(Q,\Lambda,\nu_{ini},\{M_{\sigma}\}_{\sigma\in\check{\Lambda}},F)$, where $\Lambda = \{\track{\sigma}{\tau}\mid \sigma\in\Sigma,\tau\in\Gamma\}$. Notice that each $M_{\sigma}$ uses 
only {\em rational} transition probabilities. In the second stage, we modify $M_{\sigma}$ so that it uses only transition probabilities of either $0$ or $1/d$ for a certain fixed positive integer $d$. 
This modification can be done in the following way.
Choose a sufficiently large positive integer $d$ so that any transition probability of $N$ can be expressed as $k/d$, where $k\in\nat$. Here, 
we wish to define another 1pfa $N'$. Assume that 
$Q=\{q_1,q_2,\ldots,q_m\}$. For each inner state $q_i$, we prepare $d$ 
new states $\{q^{(i)}_{1},q^{(i)}_{2},\ldots,q^{(i)}_{d}\}$. In scanning $\track{\sigma}{\tau}$, if there is a transition from inner state $q_i$ to $q_j$ with transition probability $p_{i,j}/d$, then, for each index  
$k\in[d]$, we assign the probability $1/d$ to a new transition from $q^{(i)}_{k}$ to $q^{(j)}_{l}$ for every index $l\in[p_{i,j}]$; we assign the probability $0$ to a transition from $q^{(i)}_{k}$ to $q^{(j)}_{l}$ for any other indices $l$. It is not difficult to show that this 1pfa $N'$ has the same acceptance/rejection probability as the original 1pfa $N$. 
 
Let $N'=(Q',\Lambda,\nu_{ini},\{N_{\sigma}\}_{\sigma\in\check{\Lambda}},F')$. 
For simplicity, we assume that all inner states in $Q'$ are enumerated in a pre-fixed total order. In this final stage, we want to define the desired 1dfa $\tilde{M}$ and the desired probability ensemble $\{\tilde{D}_n\}_{n\in\nat}$ that together simulate $N'$ using  $\{D'_n\}_{n\in\nat}$ with bounded-error probability. 
Let $\Delta = \{(\tau,k)\mid \tau\in\Gamma, k\in[d]\}$ be our new advice alphabet. For any advice string $w=(\tau_1,k_1)(\tau_2,k_2)\cdots (\tau_n,k_n)\in \Delta^n$ of length $n$,  
define  $\tilde{D}_n(w) = D'_n(\tau_1\tau_2\cdots\tau_n)/d^n$. Our 1dfa $\tilde{M}$ then behaves as follows. Note that, on scanning a symbol $\varphi = \track{\sigma}{\tau}$ in inner state $q_j$, since $N_{\varphi}$ is stochastic, 
$N'$ enters exactly $d$ different states, say,  $q'_{j,1},q'_{j,2},\ldots,q'_{j,d}$ (in the given order) with probability exactly $1/d$. Associated with this transition, we define a new transition of $\tilde{M}$ on a new input symbol $\track{\sigma}{(\tau,k)}$, where $k\in[d]$. In scanning this symbol $\track{\sigma}{(\tau,k)}$,  
$\tilde{M}$ enters  state $q'_{j,k}$ from $q_j$ deterministically. By the definitions of $\tilde{M}$ and $\tilde{D}_n$, it directly follows that the acceptance/rejection probability of $\tilde{M}$ with $\{\tilde{D}_n\}_{n\in\nat}$ equals that of $N'$ with $\{D'_n\}_{n\in\nat}$.

In conclusion, $A$ is recognized by the 1dfa $\tilde{M}$ with the randomized advice $\{\tilde{D}_n\}_{n\in\nat}$. This implies that $A$ is in $\reg/n$.  Since $A$ is arbitrary, we conclude that $\onebplin/Rlin \subseteq \reg/Rn$. 
\end{proof}

We have proven that $\onecequallin/Rlin$ and $\oneplin/Rlin$ are powerful enough to capture all languages. Even for weak families, such as $\reg$, randomized advice is more resourceful than deterministic advice.  

\begin{proposition}\label{Pal-random-reg}
$\dcfl\cap \reg/Rn \nsubseteq \reg/n$.
\end{proposition}

\begin{proof}
Consider a ``marked'' version of the language of {\em even-length palindromes}: $Pal_{\#} = \{w\#w^{R}\mid w\in\{0,1\}^*\}$ defined over the ternary alphabet $\Sigma = \{0,1,\#\}$. Clearly, 
$Pal_{\#}$ belongs to $\dcfl$ (because of the presence of the center marker $\#$). It is proven in \cite{Yam08a} that the language $Pal=\{ww^{R}\mid w\in\{0,1\}^*\}$ is outside of $\reg/n$. By a similar proof, we can show that $Pal_{\#}$ is not in $\reg/n$. What remains is to prove that $Pal_{\#}$ belongs to $\onebplin/Rlin$, which equals $\reg/Rn$ by Theorem \ref{1BPLIN-vs-REG}. 

Fix $n$ arbitrarily. 
Our probability distribution $D_n$ chooses advice strings of the form 
$y\#y^{R}$, where $y\in\{0,1\}^n$, with an equal probability. More precisely, for each advice string $y$ in $\{0,1\}^n$, we define 
$D_n(y\#y^{R}) = 2^{-n}$; for any other advice string $w$ in $\Sigma^{2n+1}$, let $D_n(w) =0$. 
Next, we describe an underlying 1PTM $M$ for $Pal_{\#}$ with the  above randomized advice $\{D_n\}_{n\in\nat}$. The machine $M$ behaves on input $\track{x}{w}$ as follows. If $|x|$ is even, then $M$ rejects the input. Assume that $|x|=2n+1$ and $w=y\#y^{R}$. 
Using this advice string $w$, $M$ rejects the input if $x$ is not of the form $x_1\#x_2$ for certain strings $x_1,x_2\in\{0,1\}^n$. 
Henceforth, let us assume that $x=x_1\#x_2$ with $x_1,x_2\in\{0,1\}^n$. The machine $M$ then computes two values $x_1\odot y\;(\mathrm{mod}\;2)$ and $x_2\odot y^{R}\;(\mathrm{mod}\;2)$ separately and  
$M$ finally checks whether $x_1\odot y =  x_2\odot y^{R}\;(\mathrm{mod}\;2)$. 
If those two values are equal, then $M$ accepts the input; 
otherwise, it rejects. 

Obviously, if $x_1=x_2^{R}$ then $M$ accepts $\track{x}{w}$ for any string $y\in\{0,1\}^n$. Otherwise, by the property of $\odot$, $M$ accepts $\track{x}{w}$ for exactly a 
half of $y$'s in $\{0,1\}^n$. 
Furthermore, if we modify $M$ and $\{D_n\}_{n\in\nat}$ to run {\em in parallel} the above procedure twice with two randomly-chosen advice strings $y_1\#y_1^{R}$ and $y_2\#y_2^{R}$ (which can be given as a single advice string of the form $\track{y_1\#y_1^{R}}{y_2\#y_2^{R}}$), then we can reduce the error probability down to $1/4$. It therefore holds that 
$Pal_{\#}\in\onebplin/Rlin$. 
\end{proof}

\begin{proposition}\label{reg-cequal-cfl}
$\reg/Rn\cap \onecequallin/lin\nsubseteq\cfl/n$. 
\end{proposition}

\begin{proof}
Recall the language $Dup=\{ww\mid w\in\Sigma^*\}$ over 
the alphabet $\Sigma=\{0,1\}$. 
An idea similar to the proof of Proposition \ref{Pal-random-reg} 
proves that $Dup$ belongs to $\reg/Rn$. 
Since $Dup\not\in \cfl/n$ \cite{Yam08a}, we obtain  $\reg/Rn\nsubseteq\cfl/n$. 
Next, we shall show that $Dup\in \onecequallin/lin$. 
Our advice function $h$ for $Dup$ marks the ``center'' of the tape; namely, $h(n) =0^{n/2-1}110^{n/2-1}$ if $n$ is even, and $h(n)=10^{n-1}$ if $n$ is odd.

Now, we want to define a 1PTM $M$ for $Dup$ with $h$. 
In the following description, $O$ denotes an all-zero matrix of an appropriate size and $I$ denotes an identity matrix. 
For convenience, we describe the behavior of $M$  
as a series of stochastic matrices $\{M_{\sigma}\}_{\sigma\in\check{\Lambda}}$, similar to a 1pfa, 
defined on four inner states  $\{q_0,q_1,q_2,q_3\}$.    
Here, we focus only on inputs of the form $\track{ww'}{h(2n)}$, where $w=w_1\cdots w_n$ and $w'=w'_1\cdots w'_n$ in $\Sigma^{2n}$. 
Initially, $M_{\cent}$ changes $q_0$ to $q_2$ with certainty. 
While reading 
a symbol $\sigma\in\{0,1\}$ appearing in the first half part of the 
string $ww'$, 
$M$  applies the following matrices: 
\[
M_{0} = \matrices{A}{O}{O}{1} \;\;\text{and}\;\; 
M_{1} = \matrices{B}{O}{O}{1}, 
\;\;\text{where}
\]
\[
A = \ninematrices{1}{0}{0}{0}{1}{0}{1/2}{0}{1/2} \;\;\text{and}\;\;
B = \ninematrices{1}{0}{0}{0}{1}{0}{0}{1/2}{1/2}.
\]
After scanning symbols in $w$, $M$ reaches the inner states $q_0$, $q_1$, and $q_2$ with probabilities $p_0=\sum_{i:w_i=0}2^{-i}$, $p_1=\sum_{i:w_i=1}2^{-i}$, and $2^{-n}$, respectively.  

When the head reaches the middle of $11$ written in the lower track, 
it applies the following matrix:
\[
M_{middle} = \matrices{2^{-(n-1)}I}{C}{O}{D}, \;\;\text{where}
\]
\[
C=\matrices{2^{-(n-1)}}{1-2^{-(n-2)}}{2^{-(n-1)}}{1-2^{-(n-2)}} \;\;\text{and}\;\; D=\matrices{2^{-(n-1)}}{1-2^{-(n-2)}}{0}{1}.
\]
This matrix makes the probabilities of entering states $q_0$, $q_1$, and $q_2$ equal to $p_02^{-(n-1)}$, $p_12^{-(n-1)}$, and $2^{-(n-1)}$, respectively. 
(The above matrix can be realized by the following head move: the head moves back to the left endmarker and returns to the end of the first half section by flipping fair coins.) 

As each symbol $\sigma$ in the second half of the string $ww'$, $M$ applies the matrices $M'_{0}=M_{1}$ and $M'_{1}=M_{0}$, 
and after reading $w'$, $M$ enters $q_0$ and $q_1$ with {\em extra} probabilities of $2^{-(n-1)}\sum_{i:w'_i=0}2^{-i}$ and  
$2^{-(n-1)}\sum_{i:w'_i=1}2^{-i}$, respectively. 

Finally, on scanning $\dollar$, if $M$ is already in $q_0$ and in $q_1$, it respectively enters a rejecting state, say,  $q_{rej}$ and an 
accepting state, say,  $q_{acc}$ with certainty; otherwise, it enters $q_{acc}$ and $q_{rej}$  with an equal probability. 

It is not difficult to show that $w=w'$ iff the probability of reaching 
$q_{acc}$ is exactly $1/2$. Therefore, $Dup$ belongs to 
$\onecequallin/lin$. 
\end{proof}

Since $\reg/n\subseteq \cfl/n$, Propositions \ref{Pal-random-reg} and \ref{reg-cequal-cfl} {\em both} yield a class separation $\reg/n\subsetneqq \reg/Rn$. Proposition \ref{reg-cequal-cfl} also yields another separation between $\cfl/n$ and $\cfl/Rn$, because obviously $\reg/Rn \subseteq \cfl/Rn$.

\begin{corollary}
$\cfl/n\subsetneqq \cfl/Rn$.
\end{corollary}

In the proof of Proposition \ref{reg-cequal-cfl}, we have shown that 
$Dup$ belongs to $\reg/Rn$. This fact helps us prove the following 
class separation as well. 

\begin{proposition}\label{random-reg-vs-onecequallin}
$\reg/Rn\nsubseteq \onecequallin/lin\cup \co\onecequallin/lin$.
\end{proposition}

\begin{proof}
Since $Dup$ is in $\reg/Rn$ (from the proof of Proposition \ref{reg-cequal-cfl}) and $\reg/Rn$ is obviously closed under 
complementation, $\overline{Dup}$ is also in $\reg/Rn$. In the proof of Theorem \ref{cequallin-vs-plin}, however, it is shown that 
$\overline{Dup}$ does not belong to $\onecequallin/lin$. 
Those two results imply that 
$\reg/Rn\nsubseteq \onecequallin/lin$. By considering their complement classes, we obtain another separation: $\co\reg/Rn\nsubseteq \co\onecequallin/lin$. It therefore follows that $\reg/Rn = \co\reg/Rn\nsubseteq \co\onecequallin/lin$, as requested.
\end{proof}

\section{Limitation of Randomized Advice}

The previous section has demonstrated a power of randomized advice; for example, we have shown that $\reg/Rn\nsubseteq \cfl/n \cup \onecequallin/lin$. By contrast, this section shall discuss a limitation of the randomized advice. In particular, we intend to show that $\cfl\nsubseteq \reg/Rn$; in short, even with a help of the randomized advice, $\reg$ cannot capture $\cfl$. This result significantly extends the previously-known separation $\cfl\nsubseteq\reg/n$ \cite{TYL04}

\begin{theorem}\label{cfl-vs-random-reg}
$\cfl\nsubseteq \reg/Rn$.
\end{theorem}

{}From this theorem, we can deduce that $\cfl/Rn$ {\em properly}  contains $\reg/Rn$ because, otherwise, $\cfl/n$ should be  included in $\reg/Rn$, contradicting the theorem.

\begin{corollary}
$\reg/Rn\subsetneqq \cfl/Rn$.
\end{corollary}

Henceforth, we shall prove Theorem \ref{cfl-vs-random-reg}. In order to 
do so, we borrow an idea from communication complexity theory because our model of randomized advice is loosely related to a model of two-party one-way communication with shared randomness. 
First, we introduce a new complexity class $\averreg/n$. (As for an introduction to average-case computational complexity theory, the reader may refer to \cite{Yam97} for instance.) The class  $\averreg/n$ consists of all {\em distributional problems} $(A,\mu)$, where $A$ is a language over an alphabet $\Sigma$ and $\mu=\{\mu_n\}_{n\in\nat}$ is a probability ensemble over $\Sigma^*$, such that there exist a 1dfa $M$, an advice function $h$, and an error bound $\varepsilon\in[0,1/2)$ satisfying the following condition: for every length $n\in\nat$, $\prob_{x\sim\mu_n}[M\left( \track{x}{h(n)}\right) = A(x)] \geq 1-\varepsilon$, where ``$x\sim\mu_n$'' means that $x$ is chosen randomly according to $\mu_n$.

\begin{proposition}\label{reg-to-average}
If $A$ belongs to $\reg/Rn$, then $(A,\mu)$ is in $\averreg/n$ for any probability ensemble  $\mu$.
\end{proposition}

Recall that our goal is to present a context-free language $A$ that does not belong to $\reg/Rn$. Toward this goal, with a help of Proposition \ref{reg-to-average}, it suffices to show that the distributional problem $(A,\mu)$ does not belong to $\averreg/n$ for a certain probability ensemble $\mu$. We shall present a simple example of such language $A$, known as a $\reg/n$-pseudorandom language \cite{Yam09a}. 
Formally, a language $L$ over an alphabet $\Sigma$ is called {\em $\reg/n$-pseudorandom} if, for every language $A\in \reg/n$ over $\Sigma$, the function $\ell(n) =_{def} \left| \frac{|(A\triangle L)\cap \Sigma^n|}{|\Sigma^n|} -\frac{1}{2} \right|$ is negligible, where $A\triangle L$ denotes the {\em symmetric difference} between $A$ and $L$ (\ie $A\triangle L = (A\setminus L)\cup(L\setminus A)$). 

\begin{lemma}\label{pseudorandom-average}
\sloppy If $A$ is $\reg/n$-pseudorandom, then $(A,\mu_{\rm uni})$ is not in $\averreg/n$, where $\mu_{\rm uni}=\{\mu_{{\rm uni},n}\}_{n\in\nat}$ is the uniform probability ensemble over $\Sigma^*$ (\ie each $\mu_{{\rm uni},n}$ is the uniform probability distribution over $\Sigma^n$).
\end{lemma}

\begin{proof}
We prove the lemma by contrapositive. Let $(A,\mu_{\rm uni})$ be any distributional problem in $\averreg/n$ over an alphabet $\Sigma$. 
Take a 1dfa $M$, an error bound $\varepsilon\in[0,1/2)$,  and an 
advice function $h$ satisfying that $\prob_{x\sim\mu_{{\rm uni},n}}[M(\track{x}{h(n)})=A(x)]\geq 1-\varepsilon$ for every length $n\in\nat$. 
For our convenience, we set $\varepsilon = \frac{1}{2} - \varepsilon'$ with $\varepsilon' >0$. Now, let us define $B = \{x\in\Sigma^*\mid M(\track{x}{h(|x|)})=1\}$ and consider the symmetric difference $B\triangle A$. It then follows that  
\[
\prob_{x\sim\mu_{{\rm uni},n}}[x\in B\triangle A]
= \prob_{x\sim\mu_{{\rm uni},n}}[M(\track{x}{h(n)})\neq A(x)]
\leq \frac{1}{2} - \varepsilon'.
\] 
Note that, since $\mu_{{\rm uni},n}$ is uniform, 
$\prob_{x\sim\mu_{{\rm uni},n}}[x\in B\triangle A]$ equals $|(B\triangle A)\cap\Sigma^n|/|\Sigma^n|$. {}From the above bound, the value $\ell(n)$ can be lower-bounded as  
\[
\ell(n) = \left| \frac{|(B\triangle A)\cap\Sigma^n|}{|\Sigma^n|} - \frac{1}{2} \right| 
= \left|\prob_{x\sim\mu_{{\rm uni},n}}[x\in B\triangle A] - \frac{1}{2} \right| 
\geq \varepsilon'.
\] 
This means that $A$ cannot be $\reg/n$-pseudorandom. 
\end{proof}

With a use of Proposition \ref{reg-to-average} and Lemma \ref{pseudorandom-average}, it becomes rather an easy task to prove 
Theorem \ref{cfl-vs-random-reg}, since we already know from \cite{Yam09a} that the context-free language $IP_{*}$ is in fact $\reg/n$-pseudorandom.

\begin{proofof}{Theorem \ref{cfl-vs-random-reg}}
Assume that $\cfl\subseteq \reg/Rn$. By Proposition \ref{reg-to-average}, for every context-free language $A$, the distributional problem $(A,\mu_{\rm uni})$ belongs to $\averreg/n$, where $\mu_{\rm uni}$ is the uniform probability ensemble. Lemma \ref{pseudorandom-average} further implies that $A$ cannot be $\reg/n$-pseudorandom. In summery, no  context-free language is $\reg/n$-pseudorandom. This contradicts the fact that $IP_{*}$ is a $\reg/n$-pseudorandom context-free language \cite{Yam09a}. Therefore, it should hold that $\cfl\nsubseteq \reg/Rn$.
\end{proofof}

To close this section, we still need to prove 
Proposition \ref{reg-to-average}. Its proof, in fact, follows 
immediately from a new characterization of $\reg/Rn$ given below. This characterization is a direct  consequence of Yao's principle\footnote{Yao's principle is a randomized algorithmic interpretation of von Neumann's \cite{New28} celebrated {\em minmax theorem} in game theory.} \cite{Yao77} and 
it is, to some extent, analogous to an existing result on one-way communication with public coins. 

\begin{lemma}\label{minmax-average}
Let $A$ be any language over an alphabet $\Sigma$. The following two statements are equivalent. 
\begin{enumerate}
\item $A$ is in $\reg/Rn$.

\item There exist a 1dfa $M$, an advice alphabet $\Gamma$, and an error bound $\varepsilon\in[0,1/2)$ that satisfy the following condition: for every probability ensemble $\{\mu_n\}_{n\in\nat}$ over $\Sigma^*$, there exists an advice function $h:\nat\rightarrow\Gamma^*$ such that $\prob_{x\sim\mu_n}[M(\track{x}{h(n)})=A(x)]\geq 1-\varepsilon$
for every length $n\in\nat$.
\end{enumerate}
\end{lemma}

\begin{proof}
(1 $\Rightarrow$ 2)
Let $A\in\reg/Rn$ over $\Sigma$, witnessed by $M$, $\varepsilon$, 
$\Gamma$, 
and $\{D_n\}_{n\in\nat}$ over $\Gamma^*$; that is, for every length $n\in\nat$ and for every input $x\in\Sigma^n$, it holds that 
$\prob_{y\sim D_n}[M\left(\track{x}{y}\right) \neq A(x)]\leq \varepsilon$. 
Assuming that $\{\mu_n\}_{n\in\nat}$ is a probability ensemble over $\Sigma^*$, we claim that 
\begin{quote}
\begin{itemize}
\item[(*)] for every length $n\in\nat$, there exists an advice string $y_n\in\Gamma^n$ satisfying that  $\prob_{x\sim\mu_n}[M\left(\track{x}{y_n}\right)\neq A(x)] \leq \varepsilon$. 
\end{itemize}
\end{quote}
Using this special string $y_n$, we can define the desired advice function $h$ as $h(n)=y_n$ for each length $n\in\nat$. This clearly yields (2).

Let us prove Statement (*). Assume otherwise; namely, for a certain fixed $n\in\nat$,  $\prob_{x\sim\mu_n}[M\left(\track{x}{y}\right)\neq A(x)] > \varepsilon$ for every string $y\in\Gamma^n$. 
Consider the product distribution $\nu(x,y) = \mu_n(x)D_n(y)$ for any 
pair  $(x,y)\in\Sigma^*\times\Gamma^*$. We shall estimate the value $\gamma_n = \prob_{(x,y)\sim\nu}[M\left(\track{x}{y}\right)\neq A(x)]$ in two different ways. This value $\gamma_n$ is upper-bounded as  
\begin{eqnarray*}
\gamma_n 
&=& \sum_{x\in\Sigma^n}\mu_n(x)\cdot \prob_{y\sim D_n}[M\left(\track{x}{y}\right)\neq A(x)] 
\;\;\leq\;\; \sum_{x\in\Sigma^n}\mu_n(x)\cdot \varepsilon 
\;\;=\;\; \varepsilon.
\end{eqnarray*}
However, the same value $\gamma_n$ is lower-bounded as  
\begin{eqnarray*}
\gamma_n 
&=& \sum_{y\in\Gamma^n}D_n(y)\cdot \prob_{x\sim \mu_n}[M\left(\track{x}{y}\right)\neq A(x)]
\;\;>\;\; \sum_{y\in\Gamma^n}D_n(y)\cdot \varepsilon 
\;\;=\;\; \varepsilon.
\end{eqnarray*}
The above two bounds clearly lead to a contradiction. Therefore, Statement (*) must hold. 

(2 $\Rightarrow$ 1)
Assume that there exists a 1dfa $M$, an alphabet $\Gamma$, and a constant $\varepsilon$ such that, for every probability ensemble $\mu=\{\mu_n\}_{n\in\nat}$, there exists  an advice function  $h_{\mu}:\nat\rightarrow\Gamma^*$ satisfying $\prob_{x\sim\mu_n}[M(\track{x}{h_{\mu}(n)})=A(x)]\geq 1-\varepsilon$. 
Fix $n\in\nat$ arbitrarily. 
Let us consider the following {\em two-player zero-sum game}.  
\begin{quote}
Player 1 chooses 
$y$ in $\Gamma^n$ and Player 2 chooses $x$ in $\Sigma^n$ randomly 
according to $\mu_n$. Player 1's payoff $P_{x,y}$ is $1$ if $M(\track{x}{y})=A(x)$, and $0$ otherwise. 
\end{quote}
When Player 1 tries to maximize his payoff and Player 2 tries to minimize his own payoff, we obtain the inequality 
\[
\min_{\mu_n}\max_{y}\sum_{x}\mu_n(x)P_{x,y} \geq \min_{\mu_n} \prob_{x\sim\mu_n}[M(\track{x}{h_{\mu}(n)}) = A(x)] \geq 1-\varepsilon.
\]
By Yao's principle \cite{Yao77}, it follows that 
\[
 \max_{\rho_n}\min_{x}\sum_{y}\rho_n(y)P_{x,y} = \min_{\mu_n}\max_{y}\sum_{x}\mu_n(x)P_{x,y} \geq 1-\varepsilon,
\]
where $\rho_n$ is a  probability distribution over $\Gamma^n$. 
We choose a particular $\rho_n$ that satisfies the above equality, and we  
define $D_n(y)=\rho_n(y)$ for each string $y\in\Gamma^n$. 
By the choice of $\rho_n$, for every string $z\in\Sigma^n$, we obtain 
\[
\prob_{y\sim D_n}[M(\track{z}{y})=A(z)] 
\geq \min_{x}\sum_{y}\rho_n(x)P_{x,y} \geq 1-\varepsilon.
\]
Therefore, with bounded-error probability, $M$ recognizes $A$ using 
the randomized advice $\{D_n\}_{n\in\nat}$. This means that $A$ 
belongs to $\reg/Rn$.  
\end{proof}

\section{Brief Discussion}

Throughout this paper, we have shown strengths and weaknesses 
of deterministic and randomized advice when it is given particularly 
to weak models of one-tape linear-time Turing machines and finite 
automata. 
Such weak models have made it possible to prove collapses and separations among advised language families, as shown in Figure~\ref{fig:hierarchy}, with no unproven assumption. 

Many class separations that have been proven so far are obtained 
in fact by discriminating 
two machines' abilities to extract key information from a given piece of advice. However, 
there are still numerous open questions, which we need much more 
sophisticated arguments to solve. For instance, we would like to prove/disprove the following: $\onecequallin/lin\nsubseteq \reg/Rn$, $\cfl/n\nsubseteq \onecequallin/lin$, $\reg/Rn\nsubseteq \oneplin/lin$, 
and $\oneplin/lin\nsubseteq \reg/Rn$ (even more strongly, $\oneplin/lin\nsubseteq\cfl/Rn$). 

In other research directions, one of the  
challenging tasks is to prove/disprove that $\cfl(k)/Rn\nsubseteq \cfl(k+1)/n$ for each index $k\in\nat^{+}$, where $\cfl(k) = \{\bigcap_{i=1}^{k}L_i\mid i\in[k],L_i\in\cfl\}$ \cite{Yam09a}. 
Moreover, a thorough investigation on $\averreg/n$ is certainly 
another important challenge, which may expand an existing scope of formal language and automata theory. 

\bs
\paragraph{Acknowledgments.}

The author is grateful to the Mazda Foundation and the Japanese 
Ministry of Education, Science, Sports, and Culture for their constant 
support during his research.


\end{document}